\documentclass[11pt,oneside]{article} 

\usepackage{a4wide}
\usepackage[title]{appendix}

\RequirePackage{amsthm}
\RequirePackage{stmaryrd} 
\usepackage{xcolor}
\usepackage{amsmath,amsfonts,amssymb}
\usepackage{inconsolata} 

\RequirePackage{subcaption}

\makeatletter
\newcommand{\verbatimfont}[1]{\renewcommand{\verbatim@font}{\ttfamily#1}}
\makeatother

\newcommand{\GL}{\mathrm{GL}}

\newcommand{\Complex}{\mathbb{C}}
\newcommand{\Rational}{\mathbb{Q}}

\newcommand{\Integer}{\mathbb{Z}}

\newcommand{\Field}{\mathbb{F}}

\newcommand{\Hilb}{\mathcal{H}}

\newcommand{\tensor}{\otimes}

\usepackage{mathtools}

\newcounter{numitem}

\newcommand{\ket}[1]{|{#1}\rangle}

\newcommand{\bra}[1]{\langle{#1}|}

\newcommand{\Set}{\mathsf{Set}}

\newtheorem{lemma}{Lemma}[section]
\newtheorem{proposition}[lemma]{Proposition}
\newtheorem{corollary}[lemma]{Corollary}
\newtheorem{theorem}[lemma]{Theorem}

\newtheorem{conjecture}[lemma]{Conjecture}

\theoremstyle{remark}
\newtheorem{example}[lemma]{Example}

\newcommand{\igc}[1]{\raisebox{-0.45\height}{\includegraphics{images/#1.pdf}}}

\newcommand{\todo}[1]{}

\newcommand{\CHat}{\widehat{\Complex}}
\newcommand{\CProj}{\mathbb{P}}
\newcommand{\Proj}{\mathcal{P}}
\newcommand{\PBot}{\Proj_\bot}
\newcommand{\CVec}{\Complex\text{FVec}}
\newcommand{\CMan}{\text{CMan}}
\renewcommand{\Set}{\text{Set}}
\newcommand{\PSet}{\text{Set}_\bot}
\newcommand{\PMSet}{\text{Set}_{\bot,\wedge}}

\newcommand{\Pauli}{\text{P}}
\newcommand{\Clifford}{\text{Cliff}}

\newcommand{\Matrix}[1]{\begin{pmatrix}#1\end{pmatrix}}

\title{Meromorphic Quantum Computing}

\author{Simon Burton, Hussain Anwar \\
{\small Quantinuum}
}

\begin{document}
\maketitle


\begin{abstract}
We consider the kinematic axioms of quantum mechanics projectively.
Instead of normalized (pure) states up to global phase, 
states become one-dimensional subspaces of vector spaces.
This process of projectivization is functorial and lax monoidal.
For qubits it identifies the Bloch sphere with the Riemann sphere.
We interpret a fragment of the ZXW-calculus projectively and thereby provide an
alternate derivation of the arithmetic GHZ/W-calculus of Coecke et al.
We find meromorphic functions that characterize the coherent behaviour of 
circuits for logical state preparation of quantum codes
and magic state distillation. 
\end{abstract}


\section{Introduction or: How I stopped worrying and
learnt to love the Bloch sphere}

To a mathematician a \emph{qubit} is the complex
projective line or \emph{Riemann sphere}.
To a physicist, this sounds like crazy-talk, or at best
a gimmick. What is going on?
It turns out there are good reasons for this language,
although it may not be immediately obvious. 
We end up with a description of quantum states and operators that 
looks quite different to the textbook linear-algebra version,
and we spend some time to dwell on the differences.
This formulation of quantum mechanics using projective geometry is 
found in the literature \cite{Gharahi2020,Sumadi2025,Sanchezsoto2026}
which we review below.

The usual axioms of quantum mechanics (see \cite{Nielsen2010} \S 2.2)
begin with defining a \emph{state space} as a Hilbert space $\Hilb$,
and (pure) \emph{states} as being unit vectors therein.
Moreover, we are only interested in states up to \emph{global phase}.
Unit vectors $v\in\Hilb$ modulo global phase $\phi\in U(1)$
are in one-to-one bijection with 1-dimensional subspaces of $\Hilb$.
This is called the \emph{projective space} of $\Hilb$, denoted
as $\Proj(\Hilb)$:
$$
\Proj(\Hilb) := \{ v\in \Hilb-\{0\} \} / \sim
\text{ with } u\sim \lambda u, \lambda\in\Complex-\{0\}. 
$$
These are compact manifolds.
The zero Hilbert space $0$ is sent to the empty manifold $\Proj(0)=\phi$.

For our purposes we will restrict to finite dimensional
complex vector spaces $\CVec$.
The projective space of the $d$-dimensional space $\Complex^d$
we write as $\CProj^{d-1} := \Proj(\Complex^d)$.
It follows that any invertible linear
map $T\in\GL(d,\Complex)$ respects the equivalence
relation $\sim$ and therefore acts on $\CProj^{d-1}$.

\begin{proposition}
Projectivization $\Proj$ is a functor from 
the groupoid of finite dimensional complex
vector spaces and invertible
linear maps to the category of compact manifolds $\Proj:\mathcal{GL}(\Complex)\to \CMan$.
\end{proposition}


A common notation for elements of $\CProj^{d-1}$ is 
$[v] = [v_0 : ... : v_{d-1} ]$ 
with $v_i\in\Complex$ and
$[v_0 : ... : v_{d-1} ] = [\lambda v_0 : ... : \lambda v_{d-1}]$ 
for $\lambda\in \Complex-\{0\}.$
The zero based indexing is supposed to remind us that
we lost a degree of freedom when projectivizing.

The textbook treatment of quantum mechanics
picks out states using unit vectors as a canonical
representative, which leaves some annoying fluff
on the carpet in the form of a global phase.
In other words, 
this fails as a (unique) canonical representative.
When we projectivize instead,
we pick out canonical vectors by just choosing $v_i=1$ for some $i$.
This gets rid of all the fluff,
but now there is a hole in the carpet, corresponding to
states with $v_i=0$. 
These are states ``at infinity''.
So we have traded in our old shaggy carpet, 
but at the cost of now needing multiple carpets, each with
missing points at infinity.

A particularly nice thing happens with \emph{qubits}
which have state space $\CProj^1$. 
We can label the points $ \begin{bmatrix}z\\1\end{bmatrix}$
with a single variable $z\in \Complex$,
which leaves out only one point, 
$\infty:=\begin{bmatrix}1\\0\end{bmatrix}$.
The \emph{extended complex plane} 
is defined as $\CHat := \Complex\cup\{\infty\} = \CProj^1.$
This space is also called the \emph{Riemann sphere}.
See Fig.~\ref{fig-riemann}.

Matrices 
    $\begin{pmatrix}a&b\\c&d\end{pmatrix}\in \GL(2,\Complex)$
then act on qubits as 
$ \begin{pmatrix}a&b\\c&d\end{pmatrix}
 \begin{bmatrix}z\\1\end{bmatrix}
=
 \begin{bmatrix}az+b\\cz+d\end{bmatrix}
=
 \begin{bmatrix}\frac{az+b}{cz+d}\\1\end{bmatrix}.
$
This map $z\mapsto \frac{az+b}{cz+d}$ is known
as a fractional linear transform, or \emph{M\"obius transform}.
The point at infinity goes to $a/c$ 
if $c\ne 0$ otherwise to $\infty$.

Going forwards we will be loose with 
evaluating expressions such as $1/z$ at $z=0$, which is $\infty$, and so on.

For unitary operators $g\in U(2)$ the maps $\Proj(g)$
are rotations, $\Proj(g)\in SO(3)$.
Below we use the notation $g(z) := \Proj(g)(z)$ for
 $z\in \CHat$.


\begin{figure*}[t]
\centering
\begin{subfigure}[b]{0.40\textwidth}
\includegraphics[scale=0.16]{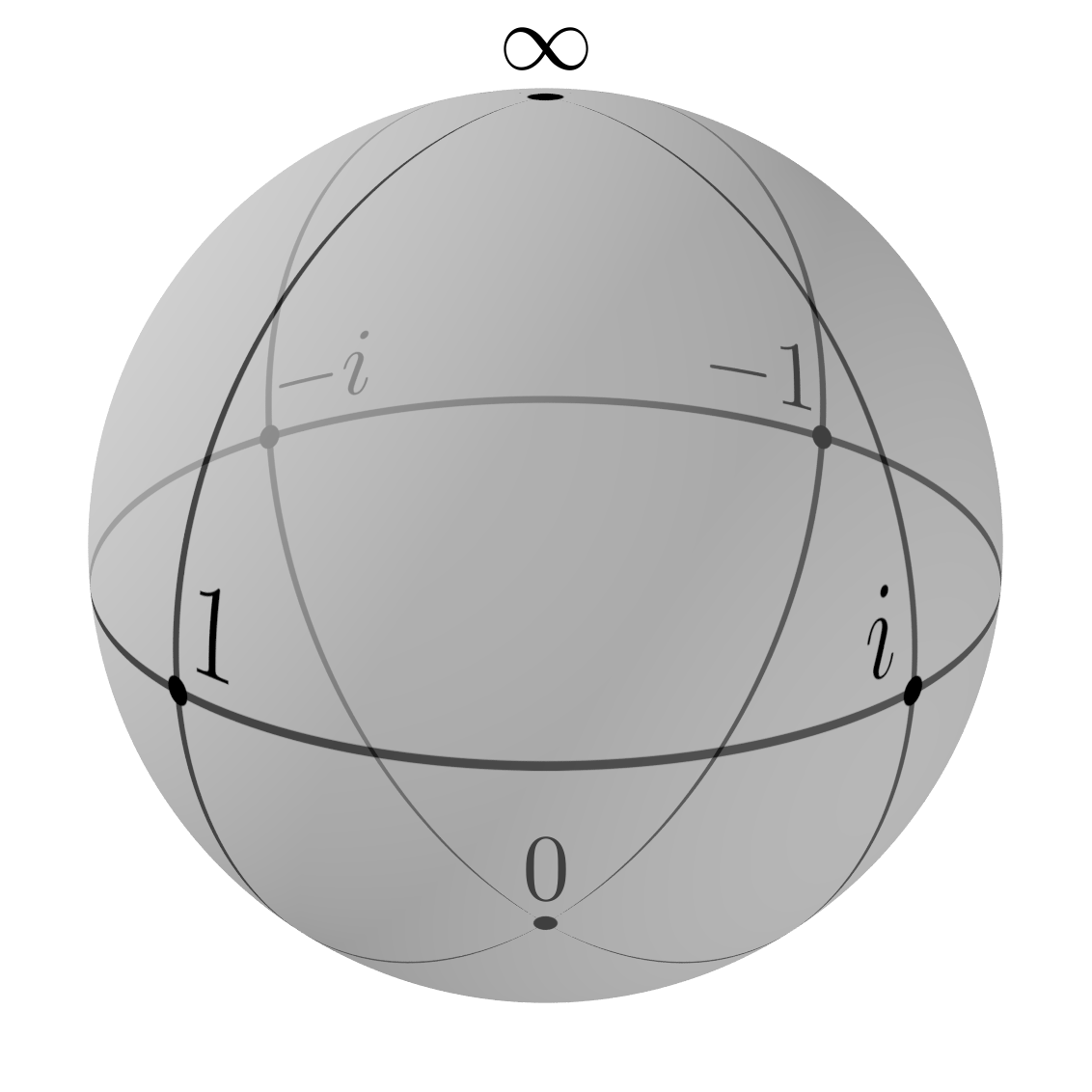}
\caption{The Riemann sphere: states are labelled by
a point $z$ in the extended complex plane.}
\label{fig-riemann}
\end{subfigure}
\begin{subfigure}[b]{0.05\textwidth}
\ \ \ \ 
\end{subfigure}
\begin{subfigure}[b]{0.40\textwidth}
\ \ \ \ 
\raisebox{10pt}{\includegraphics[scale=1.0]{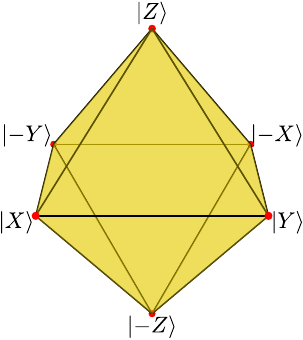}} 
\caption{The Bloch octahedron: states are 
labelled by +1 eigenvectors of Pauli matrices.}
\end{subfigure}
\caption{Two ways to label qubit states.}
\label{fig-qubit}
\end{figure*}

%

Once we apply $\Proj$ to a linear automorphism,
we will lose information about the complex eigenvalues,
but not the eigenvectors.

\begin{proposition}
Eigenvectors of linear operators $T\in\GL(n,\Complex)$
correspond bijectively to fixed points of $\Proj(T):\Proj(\Complex^n)\to \Proj(\Complex^n)$.
\end{proposition}


This proposition is slightly confusing in the case of qubits:
the projective fixed points are 
complex values $z\in\CHat$, but these correspond to eigen\emph{vectors}
of the linear operator. We have forgotten the eigen\emph{value}.


In the following table we show \emph{Pauli matrices} $g\in U(2)$
together with the action on $\CHat$ and the +1 eigenvectors $\ket{g}$ as
fixed points $g(z)=z$:
$$
\begin{array}{|c|c|c|c|}
\hline
g \in U(2)  & z \mapsto g(z) & \ket{\pm g} & g(z)=z \\
\hline
\hline
Z = \Matrix{1&0\\0&-1}  &  Z(z) = -z & \ket{\pm Z}  &  z= 0, \infty \\
\hline
X = \Matrix{0&1\\1&0}  &  X(z) = 1/z & \ket{\pm X}  &  z=\pm 1  \\
\hline
Y = \Matrix{0&-i\\i&0}  &  Y(z) = -1/z & \ket{\pm Y}  &  z=\pm i  \\
\hline
\end{array}
$$
This correspondence between points of the Riemann sphere and
eigenvectors of operators is illustrated in Fig.~\ref{fig-qubit}.
These six octahedral vertices label the \emph{stabilizer states}.
The matrices $\{ X,Y,Z \}$ generate the \emph{single qubit Pauli group}
$\Pauli_1$ of order sixteen. 
The image $\Proj(\Pauli_1)$ is a group of order four isomorphic
to the Klein-four group $\Integer_2\times\Integer_2.$
As well as the identity rotation,
these are the order two rotational symmetries of the Bloch octahedron
that fix opposite vertices.

In total there are twenty-four rotational symmetries of the octahedron,
and this group is isomorphic to $S_4$.
These rotations all come from the image $\Proj(\Clifford_1)$ of the 
\emph{single qubit Clifford group} generated by
matrices $\{S,H\}$ where
$$
S=\Matrix{1&0\\0&i}, \ \ \ 
H=\frac{1}{\sqrt{2}}\Matrix{1&1\\1&-1}.
$$
These act on the Riemann sphere as
$$
S(z) = -iz, \ \ \ 
H(z) = \frac{z+1}{z-1}.
$$
The complete list of non-trivial rotations that fix octahedral-edges follows,
along with fixed points $z\in\CHat$ and minimal polynomial $p(z)\in\Integer[z]$
for those fixed points:
$$
\begin{array}{c|l|l}
g(z)=z  & \ket{g} &  p(z)  \\
\hline
1\pm\sqrt{2}  & \ket{\pm H}  &  z^2-2z-1 \\
-(1\pm\sqrt{2})  & \ket{\pm XZH}   & z^2+2z-1   \\
+i(1\pm\sqrt{2}) & \ket{\pm ZSHS} & z^4 + 6z^2 + 1  \\
-i(1\pm\sqrt{2}) & \ket{\pm YSHS} & z^4 + 6z^2 + 1  \\
\pm(1+i)/\sqrt{2}  & \ket{\pm 1^{3/8} XS }   & z^4 + 1   \\
\pm(1-i)/\sqrt{2}  & \ket{\pm 1^{3/8} YS }   & z^4 + 1   \\
\end{array}
$$
We call any of these fixed points \emph{$H$-states}.

The non-trivial rotations that fix an octahedral-face 
are $F := SH$ and $XF, ZF, YF$. 
The fixed points, one at the center of each face,
are $z=\pm (1\pm\sqrt{3})(1\pm i)/2$.
These are the eight complex roots of the polynomial
$$
z^{8} + 14z^{4} + 1 = (z^4 - 2z^3 + 2z^2 + 2z + 1)  (z^4 + 2z^3 + 2z^2 - 2z + 1).
$$
We call any of these eight fixed points \emph{$F$-states}.

The stabilizer states, $H$-states and $F$-states are also shown
in Fig.~\ref{klein} below.

\subsection{The impossible outcome}

For us, measurement in quantum mechanics will be axiomatised
via postselection maps.
These are rank degenerate linear maps $\Pi$,
which have non-trivial nullspace in general.
The zero vector  $\Pi(v)=0$
is interpreted as an impossible outcome of the measurement.
However none of the projective spaces have an element
representing the zero vector, so we cannot
just apply $\Proj$ to $\Pi$.

The solution we take is to adjoin (disjointly) an \emph{impossible}
element $\bot$ to all our projective spaces:
$$
    \Proj_\bot(\Complex^n) = \Proj(\Complex^n) \cup \{\bot\}.
$$
Instead of topologizing this pointed space, we just
take these \emph{pointed sets}. This category
is denoted $\PSet$: objects are sets with a distinguished
element and \emph{pointed functions} that preserve the distinguished element.
Given a linear map $T:\Complex^n\to\Complex^m$, we define
the pointed function:
$$
    \Proj_\bot(T)(\bot) = \bot, \ \ \ 
    \Proj_\bot(T)([v]) = \left\{ 
\begin{array}{cl}[T(v)]&\text{ if }T(v)\ne 0,\\\bot&\text{ otherwise.}\end{array} \right.
$$

\begin{proposition}
We have a functor $\Proj_\bot:\CVec\to\PSet$
from the category of finite dimensional complex vector spaces
to the category of pointed sets.
\end{proposition}

We also think of pointed functions as partially defined functions.
The \emph{domain of definition} of a pointed function
$f:A\cup\{\bot\}\to B\cup\{\bot\}$ is the set $\{x\in A | f(x)\ne \bot\}$,
and $f(z)$ is \emph{undefined} at $z$ when $f(z)=\bot.$


%


\subsection{And}

The single qubit statespace $\CProj^1$ is a one-dimensional
complex manifold: 
it has one (complex) degree of freedom.
Therefore, preparing two qubits separably (classically),
is the product space $\CProj^1\times\CProj^1$ which has two
degrees of freedom.
The whole two qubit statespace is $\CProj^3=\Proj(\Complex^4)$,
and these states have three degrees of freedom, such as $[u:v:w:1]$.
Therefore, the two qubit separable statespace is a subspace of the
two qubit statespace.
The embedding $\CProj^1\times\CProj^1\rightarrowtail\CProj^3$ is
given by the projectivization of the tensor product
$([u:v],[w:z])\mapsto [uw:uz:vw:vz]$.
In general this is called the \emph{Segre  embedding} (see \cite{Gharahi2020}):
\begin{align*}
    \Sigma_{a,b} : \CProj^{a-1} \times \CProj^{b-1} &\rightarrowtail \CProj^{ab-1}\\
\bigl( [v], [w] \bigr) = 
\bigl( [v_0:...:v_{a-1}], [w_0:...:w_{b-1}] \bigr)
&\mapsto [v_0 w_0 : v_0 w_1 : ... : v_{a-1}w_{b-1} ]
= [v\tensor w].
\end{align*}
We will also need higher arity Segre embeddings, $\Sigma_{a,b,c,...}$
which are defined straightforwardly, as well as the unit
Segre embedding $\Sigma_a:\CProj^{a-1}\to\CProj^{a-1}$ which is just the identity.

The cartesian category of projective spaces
we are taking is just sets with cartesian product: $\Set_{\times}$.
When we include the impossible element, we must take monoidal
product to be the \emph{smash product}, written as $\wedge$. 
For pointed sets $A\cup\{\bot\}$ and 
$B\cup\{\bot\}$ their smash product is
$(A\times B)\cup\{\bot\}$.
The monoidal unit is the two element set $1_\wedge=\{\star\}\cup\{\bot\}$.
The smash product of pointed functions  $f,g$ is
$$
(f\wedge g)(\bot) = \bot, \ \ \ 
(f\wedge g)(a, b) = \left\{ 
\begin{array}{cl}(f(a),g(b))&\text{ if }f(a)\ne\bot\text{ and }g(b)\ne\bot,
\\\bot&\text{ otherwise.}\end{array} \right.
$$
The monoidal category of pointed functions
and smash product is written $\Set_{\bot,\wedge}$.

Extending the Segre embedding to pointed projective spaces
is also straightforward. The pointed Segre embedding,
also written as $\Sigma$, is the map:
\begin{align*}
\CProj^{a-1}_\bot \wedge \CProj^{b-1}_\bot &\to \CProj^{ab-1}_\bot \\
\bot,\ ([u],\bot),\ (\bot,[v]) &\mapsto \bot \\
([u], [v]) &\mapsto \Sigma_{a,b}([u],[v]).
\end{align*}

\begin{proposition}\label{lax}
Using the pointed Segre embedding as laxator, 
$\Proj_\bot$ becomes a lax monoidal functor 
$$\Proj_\bot:\CVec_{\tensor} \to \text{Set}_{\bot,\wedge}.$$
\end{proposition}
\begin{proof}See App.~\ref{applax}.\end{proof}


\subsection{Summary}

In summary, the kinematic axioms of projective quantum mechanics are:

(1) states are points in projective spaces, or otherwise impossible

(2) operators are maps of projective spaces that preserve the impossible state  

(3) composition of states is lax monoidal 

\section{Meromorphic functions}

A \emph{meromorphic function} on the Riemann sphere is
a fractional linear transform
$$
    f(z) = \frac{p(z)}{q(z)}
$$
where $p,q\in\Complex[z]$ are polynomials in $z$ with coefficients in $\Complex$, 
without any roots in common,
and $q\ne 0$.
\footnote{
At times we will forget the distinction between 
formal elements such as $p\in\Complex[z]$ and functions $p:\Complex\to\Complex$.}
This set of meromorphic functions is the field
of \emph{rational functions in one variable}, denoted $\Complex(z)$.
(See \cite{Girondo2012} Prop. 1.26.)
The \emph{degree} of $f$ is 
the maximum of the degrees of $p$ and $q$.  (See \cite{Girondo2012} Ex. 1.58.)
This counts, with multiplicity,
the number of points in any fiber $f^{-1}(w)$ over $w\in\CHat$.
A \emph{branch point} of $f$ is $z\in\CHat$ such that $f'(z)=0$.
These are the stationary points of $f$.
The multiplicity of such a branch point is the
``number of sheets that meet at $z$'', written $m_z(f)$.
In practice, 
this is one plus the order of the zero of $f'(z)$.
To find all the stationary points
sometimes we need to change variables, such as $z\mapsto 1/z$,
in order to deal with the point at infinity.

The following theorem is a baby version of the Riemann-Hurwitz theorem (see \cite{Girondo2012} Thm. 1.60.):
\begin{theorem}\label{RH}
Given $f\in\Complex(z)$ with degree $d$ and branch points $B\subset \CHat$, we have
$$
    \sum_{z\in B} (m_z(f) - 1) = 2(d-1).
$$
\end{theorem}

We also use \emph{multivariate meromorphic functions}
$f(z_1,...,z_n)$ which are defined to be quotients 
$p/q$ of polynomials $p,q\in\Complex[z_1,...,z_n]$ without
common zeros, and $q\ne 0$.

\begin{proposition}\label{mero}
Given a linear operator $T:(\Complex^2)^{\tensor n}\to \Complex^2$,
the function $f:\bigtimes_n\CProj^{1} \to \CProj^1$ given by
$f=\Proj_\bot(T)\Sigma_{2,...,2}$ is a meromorphic function of $n$ variables
on its domain of definition:
$$
f(z_1,...,z_n) = \left\{\begin{array}{l}
\frac{p(z_1,...,z_n)}{q(z_1,...,z_n)} \\
\\
\bot
\end{array}\right.
$$
\end{proposition}
\begin{proof}
Applying the Segre embedding to 
the $n$-qubit state $([z_1:1], ..., [z_n:1])$ 
gives $[z_1z_2...z_n : ... : 1]$,
where each entry is a monomial in the $z_i$ or $1$.
The action of $T$ on this (projective) vector is
$[p(z_1,...,z_n):q(z_1,...,z_n)]=[\frac{p(z_1,...,z_n)}{q(z_1,...,z_n)}:1]$
where $p,q\in\Complex[z_1,...,z_n]$ with common zeros eliminated.
\end{proof}

\subsection{Klein's Octahedral function}

We define a meromorphic function,
the \emph{octahedral function} as
$$
E_7(z) := \frac{108 z^4(z^4-1)^4}{(z^8 + 14z^4 + 1)^3}
$$
This is used to detect when points are related
by a Clifford rotation.
We could also think of this as a coequalizer of the image of
$\Clifford_1$ under $\Proj$.

\begin{proposition}\label{E7}
Given $w,z\in\CHat$, then $w\in\text{Orbit}_{\Clifford_1}(z)$
if and only if $E_7(w)=E_7(z)$.
\end{proposition}
The proof of this proposition goes back to Klein's book 
from 1888 \cite{Klein2003} on invariant theory and
is connected to the modern theory of modular curves~\cite{Zvonkin2008},
and dessin d'enfants~\cite{Girondo2012}.
See Fig.~\ref{klein}.

\begin{figure*}[t]
\centering
\begin{subfigure}[b]{0.40\textwidth}
\includegraphics[scale=0.16]{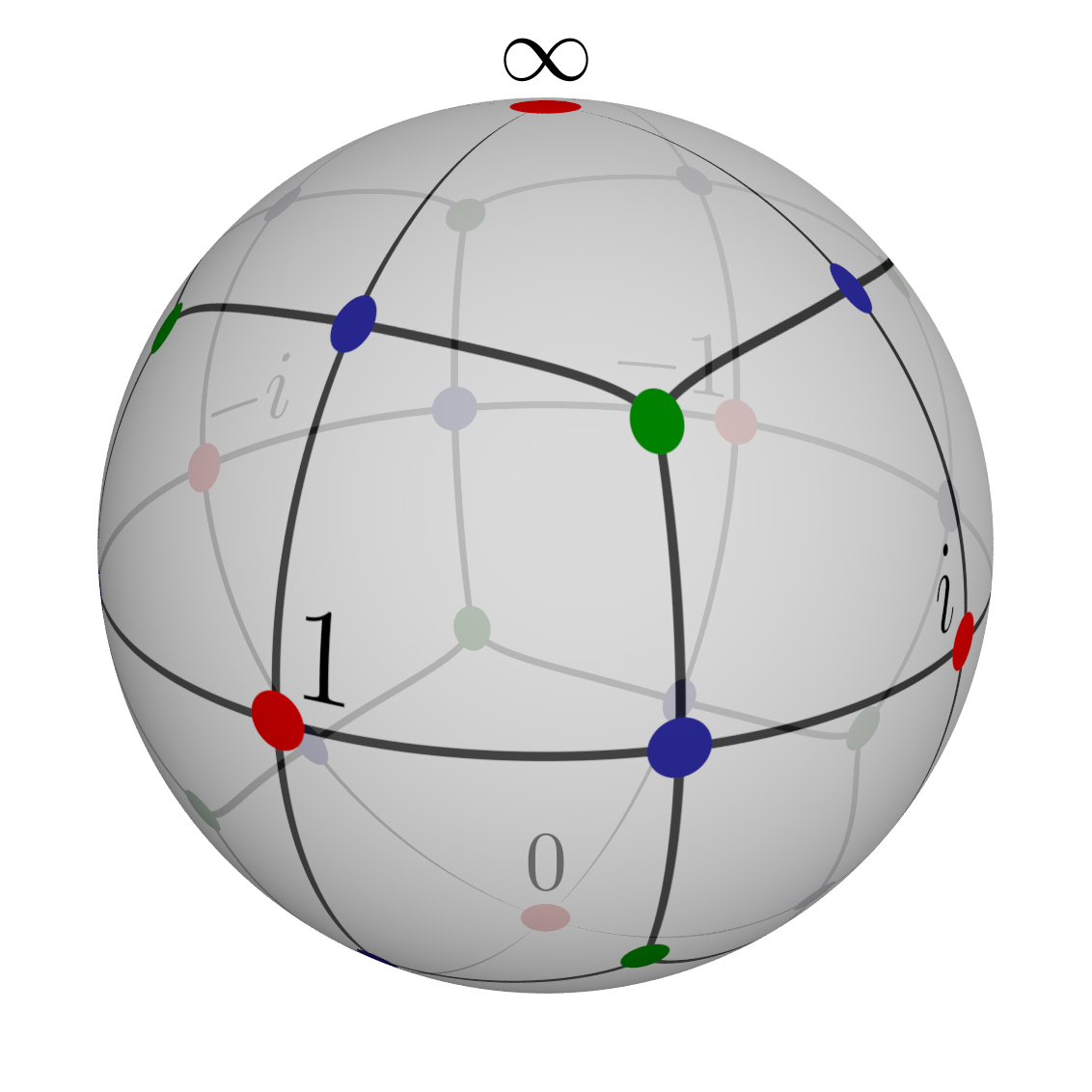}
\caption{The Riemann sphere $\CHat=\CProj^1$ with the six stabilizer states
in red, twelve $H$-states in blue, 
and eight $F$-states in green.}
\label{klein}
\end{subfigure}
\begin{subfigure}[b]{0.05\textwidth}
\ \ \ \ 
\end{subfigure}
\begin{subfigure}[b]{0.40\textwidth}
\ \ \ \ 
\raisebox{10pt}{\includegraphics[scale=0.16]{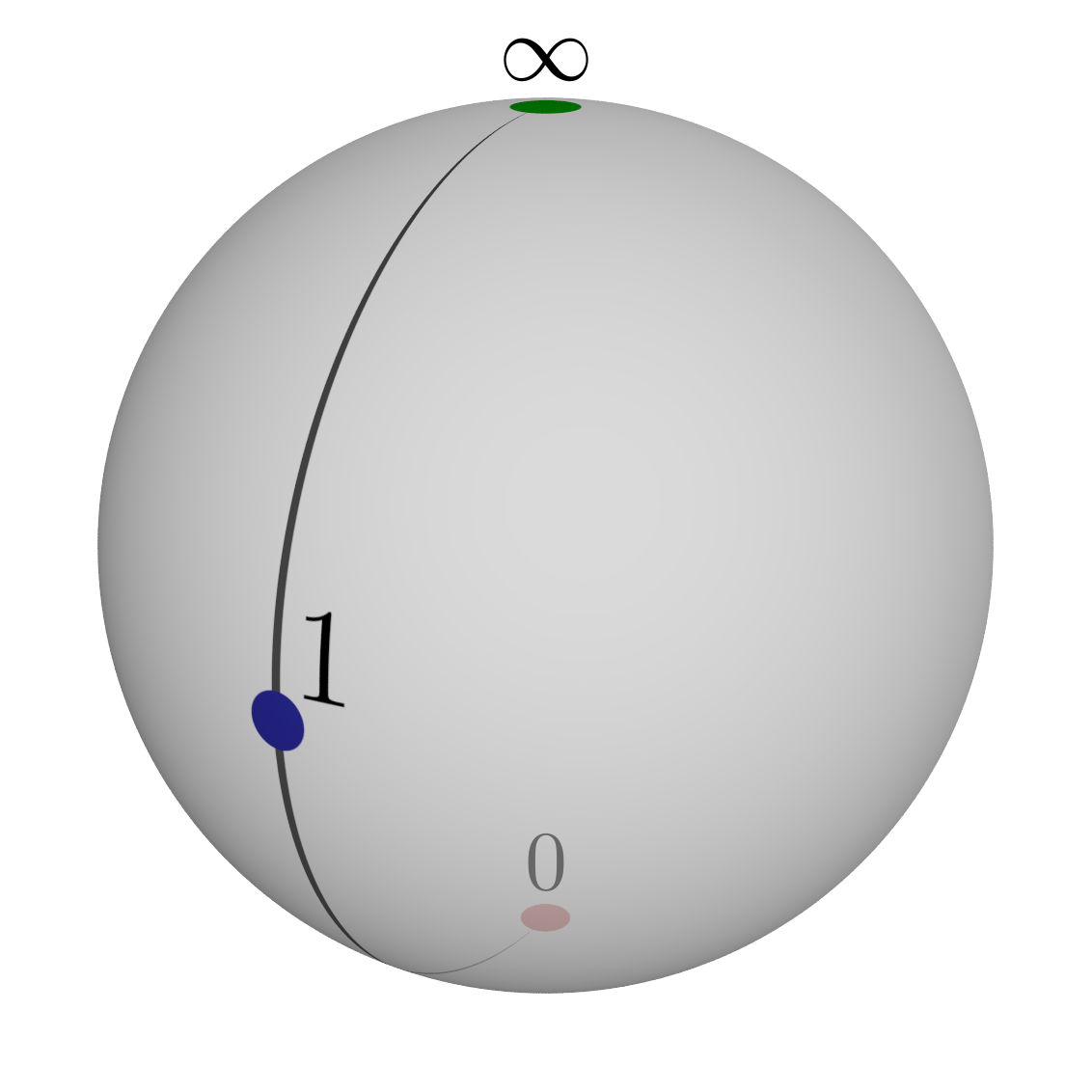}} 
\caption{The Riemann sphere with the point $z=0$ in red,
$z=1$ in blue, and $z=\infty$ in green. This is the image
of Klein's octahedral function $E_7$.}
\end{subfigure}
\caption{We show an octahedral ``design'' on the Riemann
sphere, and its image under $E_7$ on the right.
Each of the twenty-four tiles on the left is mapped to a
single region on the right which covers the sphere.
}
\label{klein}
\end{figure*}



\section{Circuits and diagrams}

We recall the tensor network notation known as $ZX$-calculus~\cite{Van2020}.
Our diagrams flow from right-to-left, in agreement with algebraic (Dirac)
notation.
We define the $Z$-type (white) or $X$-type (black) \emph{spider} with
$m$ outputs and $n$ inputs, labelled with a non-zero
\emph{multiplicative phase} $z\in\Complex/\{0\}$:
\begin{align*}
\igc{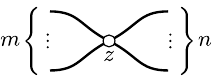}\ \ & :=\ \ \ket{+Z}^{\tensor m}\bra{+Z}^{\tensor n} + z \ket{-Z}^{\tensor m}\bra{-Z}^{\tensor n} \\
\igc{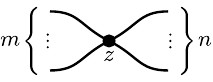}\ \ & :=\ \ \ket{+X}^{\tensor m}\bra{+X}^{\tensor n} + z \ket{-X}^{\tensor m}\bra{-X}^{\tensor n}
\end{align*}
When the multiplicative phase $z$ is omitted, it defaults to $z=1$.
We also define the following two \emph{$W$-spiders}:
$$
\igc{w_ww} :=
\begin{pmatrix}0&1&1&0\\0&0&0&1 \end{pmatrix},
\ \ \ \ 
\igc{w_} :=
\igc{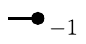} 
=
\begin{pmatrix}0\\1 \end{pmatrix}.
$$

We recall  the three structures 
$\Bigl( \igc{g_}, \igc{g_gg} \Bigr) $,
$\Bigl( \igc{r_}, \igc{r_rr} \Bigr) $ and
$\Bigl( \igc{w_}, \igc{w_ww} \Bigr) $
each give complex unital associative algebras on 
$\Complex^2$ \cite{Heunen2019,Coecke2010}.

\begin{proposition}\label{zx}
Projectivization induces the three 
(associative, unital) monoid structures on the object 
$\CProj^1_\bot$ 
in  $\Set_{\bot,\wedge}$:
\begin{center}
\begin{tabular}{|c|c|c||c|c||c|c|}
\hline
& & & & & & \\
$\CVec$ matrix & $\begin{pmatrix}1\\1 \end{pmatrix}$ &
$\begin{pmatrix}1&0&0&0\\0&0&0&1 \end{pmatrix}$ &
$\begin{pmatrix}1\\0 \end{pmatrix}$ &
$\begin{pmatrix}1&0&0&1\\0&1&1&0 \end{pmatrix}$ &
$\begin{pmatrix}0\\1 \end{pmatrix}$ &
$\begin{pmatrix}0&1&1&0\\0&0&0&1 \end{pmatrix} $
\\
 & & & & & & \\
diagram & \includegraphics{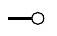} &
\includegraphics{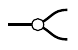} &
\includegraphics{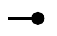} &
\includegraphics{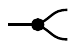} &
\includegraphics{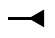} &
\includegraphics{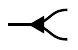}
\\
 & & & & & & \\
projectivization & $1 \mapsfrom $ &
$ \left.\begin{array}{c}wz\\\bot\end{array}\right\} \mapsfrom (w,z) $   &
$\infty \mapsfrom $ &
$ \left.\begin{array}{c}\frac{wz+1}{w+z}\\\bot\end{array}\right\} \mapsfrom (w,z) $   &
$ 0 \mapsfrom $ &
$ \left.\begin{array}{c}w+z\\\bot\end{array}\right\} \mapsfrom (w,z) $ 
\\
 & & & & & & \\
undefined at 
 & & $(\infty,0),(0,\infty)$
& & $(1,-1),(-1,1)$ & & $(\infty,\infty)$ \\
\hline
\end{tabular}
\end{center}
\end{proposition}
\begin{proof}
The formulas come from applying the matrix to projective coordinates
given by the Segre embedding.
These give monoid structures 
either by direct inspection, or by noting that 
lax monoidal functors preserve monoids.
\end{proof}

In Ref.\cite{Coecke2011} they build the field of rational numbers
formally out of the $Z$ and $W$ spider,
which forms the $ZW$-calculus \cite{Coecke2010}. 
We now see that  this comes directly from reformulating
linear algebra as projective algebraic geometry:
the multiplication and addition of complex numbers is 
the projectivization of the $Z$- and $W$-spiders.

But what about the $X$-spider?
Taking the function 
$(w,z) \mapsto \frac{wz+1}{w+z}$
and switching to the reciprocal chart $z\mapsto 1/z$
of the Riemann sphere, we find 
$$
(w,z) \mapsto \frac{w+z}{wz+1}
$$
which is the familiar formula for Lorentz velocity addition.

From \cite{Coecke2011-1}, page 28,
we see the $W$-spider is expressed 
in $ZX$-calculus as
$$
\igc{w_ww}
=
\igc{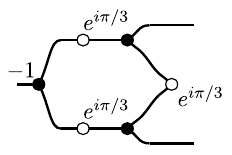}.
$$

With these tools we can now express such things
as Mandelbrot iteration $z\mapsto z^2 + c$ which belong in 
the field of ``complex dynamics'', see~\cite{Zvonkin2008}\S 6.



\section{Application to stabilizer codes}


The $n$-qubit \emph{Pauli group} denoted $\Pauli_n$ is generated
by $n$-fold tensor products $\tensor_{i=1}^{n} g_i$ with $g_i$ an
element of the (1-qubit) Pauli group.
We also write indexed Pauli group elements such as $X_i$ to mean 
a tensor product $I\tensor ... X ...\tensor I $ with $X$ in the $i$'th 
component, etc.
An \emph{$X$-type} Pauli operator is a tensor product of $I,X\in \Pauli_1$.
Similarly for \emph{$Z$-type} Pauli operators.
The $n$-qubit \emph{Clifford group} denoted $\Clifford_n$ 
is the normalizer of $\Pauli_n$ in the unitary group $U(2^n)$.
It acts on $\Pauli_n$ by conjugation.
An element $g\in\Pauli_n$ has \emph{weight} $w(g)$ which 
counts how many of the $n$ components 
in the $n$-fold tensor product $g$
are not proportional to the identity.

A (quantum stabilizer) \emph{code} \cite{Gottesman1997}
on $n$-qubits, denoted by $C$, is specified 
by an abelian subgroup $S\subset\Pauli_n$ such that $-I\notin S$.
The \emph{logical subspace} of $C$ is the linear subspace 
$\{v\in(\Complex^2)^{\tensor n} | gv=v \text{ for all } g\in S\}.$
Given independent stabilizer generators $S = \langle g_i \rangle_{i=1}^{m}$,
we find $k$ pairs of anti-commuting independent
\emph{logical operators} $\{ L_{X_i}, L_{Z_i} \}_{i=1}^{k}$,
that together with $S$ generate the centralizer 
$C_{\Pauli_n}(S)$.
An \emph{encoder} 
is a unitary $E\in \Clifford_n$ such that
$E^\dag Z_i E = g_i$ for $i=1,..,m$ and
$E^\dag X_{i+m} E = L_{X_i},
E^\dag Z_{i+m} E = L_{Z_i}$ for $i=1,...,k$.
The \emph{logical encoder} is the linear map 
$L = E(\ket{Z}^{\tensor m}\tensor I^{\tensor k})$.
The \emph{logical decoder} is $L^\dag$.
The projector onto the logical subspace is then 
$\Pi=LL^\dag $. 
The \emph{distance} of a code $d$ is the minimum weight of
a logical operator up to multiplication by a stabilizer.
We will take the label $C$ to denote any or all of this 
data as needed;
the notation $[[n,k,d]]$ is also used to refer to these parameters.


\begin{proposition}
Given a code $C$ with parameters $[[n,1,d]]$ and logical decoder $L^\dag$,
the function $f:\bigtimes_n\CProj^{1} \to \CProj^1$ given by
$f=\Proj_\bot(L^\dag)\Sigma_{2,...,2}$ is a 
meromorphic function of $n$ variables on its domain of definition.
\end{proposition}
\begin{proof}
Apply proposition ~\ref{mero} to the linear operator $L^\dag$.
\end{proof}

For simplicity of exposition, we frame the following results in terms of 
a univariate meromorphic function $f(z)=f(z,...,z)$.
We call such a function $f(z)$ the \emph{meromorphic decoder} of $C$.
We say that a meromorphic decoder $f(z)$ \emph{weakly distills}
a point $z\in \CHat$ if $z$ is a stationary point of $f(z)$.
The degree of such a stationary point is $m_z(f)$ and 
we call this \emph{coherent error suppression} to order
$O(\varepsilon^{m_z(f)})$.

\begin{theorem}\label{WD}
Given a meromorphic decoder of degree $n$ 
that weakly distill points $D\subset \CHat$ then
$$
    2(n-1) = \sum_{z\in D} (m_z(f)-1).
$$
\end{theorem}
\begin{proof}
This is a restatement of Thm.~\ref{RH}.
\end{proof}

We say a meromorphic decoder $f(z)$ 
\emph{distills} a point $z\in\CHat$ when $f(z)=z$,
and \emph{distills up to a Clifford} when 
$f(z)=g(z)$ with $g\in\Clifford_1$.
Furthermore, $f(z)$ 
\emph{coherently distills} a point $z\in\CHat$ when 
$f'(z)=0$ and also $f(z)=z$, and 
\emph{coherently distills up to a Clifford} if $f'(z)=0$
and $f(z)=g(z)$ with $g\in\Clifford_1$.

Given a code $C$ with encoder $E$,
the \emph{dual code} $C^\top$ has
encoder $H^{\tensor n}E$.


\begin{proposition}\label{dual}
Given $k=1$ code $C$ with meromorphic decoder $f$,
and dual code $C^{\top}$ with meromorphic decoder $f^{\top}$, we have
$$
f(z) = \frac{f^\top(\frac{z+1}{z-1})+1}{f^\top(\frac{z+1}{z-1})-1}.
$$
\end{proposition}
\begin{proof}
Conjugate the logical decoder $L^\dag$
with $H$ and use functoriality of $\Proj$.
\end{proof}

A CSS code $C$ has stabilizer group $S = S_X S_Z$ where
$S_Z$ has only $Z$-type Pauli operators, and $S_X$ has only $X$-type. 
Given a set of operators $A\subset \Pauli_n$, we define the 
\emph{weight enumerator} of $A$ to be the 
following homogeneous polynomial in $\Complex[x,y]$:
$$
    W_A(x,y) = \sum_{g\in A} x^{n-w(g)} y^{w(g)}. 
$$

\begin{theorem}\label{wenum}
Given CSS code $C$ with stabilizers $S_X,S_Z$ and logical operators 
$X^{\tensor n}, Z^{\tensor n}$
the meromorphic decoder of $C$ is given by
$$
f(z) = \frac{W_{\langle S_X\rangle}(z,1)}{W_{\langle S_X\rangle}(1,z)}
$$
\end{theorem}
\begin{proof}See App.~\ref{appwenum}.\end{proof}


Note that for CSS codes
we can always recover the $Z$-type stabilizers from the $X$-type
stabilizers and $X$-type logical operators, and vice-versa.

\begin{corollary}\label{macw}
Given CSS code $C$ with stabilizers $S_X,S_Z$ and logical operators 
$X^{\tensor n}, Z^{\tensor n}$
we have
$$
f(z)
=
\frac{W_{\langle S_Z\rangle}(\frac{z+1}{z-1},1) + W_{\langle S_Z\rangle}(1,\frac{z+1}{z-1})}
{W_{\langle S_Z\rangle}(1,\frac{z+1}{z-1}) - W_{\langle S_Z\rangle}(1,\frac{z+1}{z-1})}.
$$
\end{corollary}
\begin{proof}
Combine the previous proposition and theorem.
\end{proof}

Prop.~\ref{dual} and/or Cor.~\ref{macw} can be seen as a kind of
MacWilliams theorem for meromorphic decoders~\cite{Macwilliams1977}.

Given a code $C$ with $k=1$,
and another code $D$ on $n$ qubits we can use the
logical operators of $n$ copies of $C$ to express the 
stabilizers of $D$. The resulting code is the \emph{concatenated}
code $C\tensor D$.

\begin{proposition}\label{cat}
Given two codes $C$ and $D$ both with $k=1$, and meromorphic decoders 
$f$ and $g$ respectively,
the meromorphic decoder of the concatenated code $C\tensor D$ is $f\circ g$.
\end{proposition}
\begin{proof}
The logical decoders of $C$ and $D$ compose
to give the logical decoder of $C\tensor D$,
so the proposition follows from functoriality of $\Proj$.
\end{proof}

We will also define the \emph{$r$-polynomial} for a meromorphic decoder 
$f(z) = \frac{p(z)}{q(z)}$ as $r(z) = p(z) - zq(z).$
This gives an algebraic characterization of the fixed points of a distillation
routine.

\begin{proposition}\label{fixed}
The roots of $r(z)=0$ correspond bijectively to distillation points $z\in\Complex$.
\end{proposition}

\todo{define coherent error suppression  $O(\varepsilon^d)$ }

\begin{conjecture}
Any CSS code $[[n,1,d]]$ with logical operators
$X^{\tensor n}, Z^{\tensor n}$
exhibits coherent error suppression $O(\varepsilon^d)$
at the four stabilizer states $z=0,\infty,\pm 1$.
\end{conjecture}
This conjecture says that the meromorphic decoder $f(z)$ for
these codes will have order $d-1$ zeros of $f'(z)$ at $z=0,\infty,\pm 1.$
We will see this in examples below as the factor $z^{d-1}(z-1)^{d-1}(z+1)^{d-1}$
appearing in the numerator of the expression for $f'(z)$.

\todo{use Prop.~\ref{E7}}



\section{Examples}

Many of the following calculations were performed using 
sage python \cite{Sagemath}, which has powerful algorithms for polynomial
factorization and root-finding, see App.~\ref{appsage}.

In the following examples we suppress the tensor
product symbol $\tensor$ when notating elements of $\Pauli_n$,
as well as writing a dot $.$ for $I\in\Pauli_1$.

%

\subsection{Logical state preparation}

Given a code $C$,
by applying the code projector $\Pi=LL^\dag$ to $n$
physical qubits we find the encoded state to be given
by the meromorphic decoder. 

%

\begin{example}
Here we use logical decoder circuits from the 
\emph{phase-free} ZX-calculus, see \cite{Kissinger2022}.
The stabilizer code  $\langle XXI,IXX\rangle$
has logical decoder circuit  
$
L^\dag = \igc{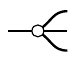}
$
from which we derive the meromorphic decoder via Prop.~\ref{zx}:
$
g(z_1,z_2,z_3) = z_1 z_2 z_3.
$
For the other stabilizer code  $\langle ZZI, IZZ\rangle$,
the logical decoder circuit is 
$$
\igc{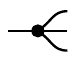}
=\igc{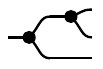}
$$
Once again using Prop.~\ref{zx} we find the meromorphic decoder,
\begin{align*}
h(z_1,z_2,z_3) =& \frac{1+z_1z_2}{z_1+z_2}
\ \ \mathrm{with}\ \ 
z_2 \mapsto \frac{1+z_2z_3}{z_2+z_3}\\
= & 
 \frac{1 + z_1 \frac{1+z_2z_3}{z_2+z_3} }
{z_1 + \frac{1+z_2z_3}{z_2+z_3} }\\
= & \frac{z_1 + z_2 + z_3+ z_1z_2z_3 }
{1 + z_1z_2 + z_1 z_3+ z_2z_3 } \\
\end{align*}
Concatenating these two codes together produces the 
$[[9,1,3]]$ Shor code:
\begin{center}
\begin{tabular}{|c|c|}
\hline
$S$ &  {\tt XX.} \\
  &  {\tt .XX} \\
\hline
$L_X$  &  {\tt XXX} \\
$L_Z$  &  {\tt ZZZ} \\
\hline
\end{tabular}
\ \ $\tensor $ \ \ 
\begin{tabular}{|c|c|}
\hline
$S$ &  {\tt ZZ.} \\
  &  {\tt .ZZ} \\
\hline
$L_X$  &  {\tt XXX} \\
$L_Z$  &  {\tt ZZZ} \\
\hline
\end{tabular}
\ \ $=$ \ \ 
\begin{tabular}{|c|c|}
\hline
$S$  &  {\tt ZZ.......} \\
  &  {\tt ...ZZ....} \\
  &  {\tt ......ZZ.} \\
  &  {\tt .ZZ......} \\
  &  {\tt ....ZZ...} \\
  &  {\tt .......ZZ} \\
  &  {\tt XXXXXX...} \\
  &  {\tt ...XXXXXX} \\
\hline
$L_X$  &  {\tt XXXXXXXXX} \\
$L_Z$  &  {\tt ZZZZZZZZZ} \\
\hline
\end{tabular}
\end{center}
For simplicity we write the univariate meromorphic decoders of 
the concatenating codes as 
$g(z)=\frac{z^3 + 3z}{3z^2 + 1}$ and
$h(z)=z^3$, respectively.
Using Prop.~\ref{cat}, the Shor code has meromorphic decoder, 
$f(z)$, $r$-polynomial, and derivative $f'(z)$:
\begin{align*}
f(z) & = g(h(z)) = \frac{z^{9} + 3 z^{3}}{3 z^{6} + 1} \\
r(z) &= (z - 1) z (z + 1) (z^{2} - z - 1) (z^{2} + 1) (z^{2} + z - 1) \\
f'(z) &= \frac{9 (z - 1)^{2} z^{2} (z + 1)^{2} (z^{2} - z + 1)^{2} (z^{2} + z + 1)^{2}}{(3z^{6} + 1)^2}
\end{align*}
The $r$-polynomial has roots at $z=0,\pm1,\pm i$
and so these stabilizer states are distilled by Prop.~\ref{fixed}.
Of these, $z=0,\pm 1$ are stationary points of degree three,
giving coherent distillation at $z=0,\pm 1$ that  
suppresses coherent error to $O(\varepsilon^3)$.
\qed \end{example}

\todo{show entire [[9,1,3]] decoder}

\subsection{Magic state distillation}

We repeat the same analysis of the meromorphic decoder, but
now we are interested in \emph{magic states}, which are  
non-stabilizer states prepared by the decoder $L^\dag$.
See \cite{Bravyi2005,Zheng2024}.

\begin{example}
Here we consider 
the original magic state distillation procedure from \cite{Bravyi2005}.
The $[[5,1,3]]$ code has stabilizer group  
$ S = \langle XZZXI,IXZZX,XIXZZ,ZXIXZ\rangle$.
The meromorphic decoder is given by 
$$
    f(z) = \frac{z^5 - 5z}{5z^4 - 1}
$$
with derivative
$$
f'(z) = \frac{5 (z^8 + 14z^4 + 1)}{(5z^{4} - 1)^2}
$$
The degree of $f(z)$ is five, and there are eight weakly distilled points,
all of degree two.
These are the eight $F$-states.
We see $2\times(5-1)=8$ which verifies Thm.~\ref{WD}.
The $r$-polynomial for this code is
$$
r(z) = -4 z (z^{4} + 1) 
$$
and so the meromorphic decoder distills $z=0$
as well as four of the $H$-states. 
Using the octahedral function, we find that $E_7(f(z))-E_7(z)=0$
has solutions at all stabilizer states, $H$-states and $F$-states.
So by Prop.~\ref{E7}, up to a Clifford correction,
this meromorphic decoder can be used to coherently 
distill $F$-states suppressing coherent error up to $O(\varepsilon^2)$.
\qed \end{example}

\begin{example}
The $[[7,1,3]]$ \emph{Steane code} has stabilizer group 
\begin{align*}
S = 
\langle 
&IIIXXXX, IXXIIXX, XIXIXIX, \\
&IIIZZZZ, IZZIIZZ, ZIZIZIZ
\rangle
\end{align*}
which is a CSS code with logical
operators $X^{\tensor 7},Z^{\tensor 7}$.
This has weight enumerator 
$W_{\langle S_X\rangle}(x,y)= x^7 + 7x^3 y^4$ and so by Thm.~\ref{wenum}:
\begin{align*}
f(z) &= \frac{W_{\langle S_X\rangle}(z,1)}{W_{\langle S_X\rangle}(1,z)} 
  = \frac{ z^{7} + 7 z^{3}}{7 z^{4} + 1} \\
r(z) &= (z - 1) z (z + 1) (z^{2} - 2z - 1) (z^{2} + 2z - 1) \\
f'(z) &= \frac{21 (z - 1)^{2} z^{2} (z + 1)^{2} (z^{2} + 1)^{2}}{(7z^{4} + 1)^2}
\end{align*}
The distillation map has degree seven,
and we see it has stationary
points only at the six stabilizer states,
each of degree three.
From the $r$-polynomial, 
it distills the four $H$-states at 
$z^{2} \pm 2z - 1=0$, but these are not stationary points,
and so $f$ does not suppress coherent errors at these points. 
\qed \end{example}

\begin{example}
For the $[[15,1,3]]$ triorthogonal code (\cite{Bravyi2005}\S VI) we find:
\begin{align*}
f(z) &= \frac{ z^{15} +  15 z^{7}}{15 z^{8} + 1 } \\
r(z) &= (z - 1) z (z + 1) (z^{12} + z^{10} + z^{8} - 14z^{6} + z^{4} + z^{2} + 1) \\
f'(z) &= 
\frac{105z^6 (z-1)^2 (z+1)^2 (z^2+1)^2 (z^4+1)^2}{(15z^8 + 1)^2}
\end{align*}
This suppresses coherent error of the
four $H$-states at $z^4+1=0$ to order $O(\varepsilon^3)$,
which are distilled up to Clifford rotation.
\todo{distillation of $z=0$ up to degree $O(\varepsilon^7)$,
refine CSS conjecture. 
There are eight more points distilled up to Clifford
which are $z=\pm 1^{5/16}, \pm 1^{3/16}, \pm(1^{5/48}-1^{13/48})$}.
\qed \end{example}



\section{Departing thoughts and questions}

We are making use of algebraic geometry ideas to reformulate
quantum computing primitives. 
By projectivizing in a functorial way, we begin to
translate the axioms of quantum mechanics into a different, non-linear language.
These are the algebraic bones of quantum computing.
We end up with something 
that is ``the same'' as the usual formulation of quantum physics,
although we would like to see a more precise statement of this. 
Apparently Heisenberg learnt about linear
algebra from Max Born~\cite{Heilbron2023}. 
Perhaps if he had wandered a bit further
down the hallway to find an algebraic geometer instead 
we would now be rewriting quantum mechanics in terms of linear algebra.

Meromorphic functions are found  in many places in the literature,
such as in signal processing and analytic combinatorics.
Stationary points of our meromorphic distillation 
functions hold the key to coherent error suppression in
(magic) state preparation, and these appear quite generically by 
Thm.\ref{WD}.

A more refined
analysis would differentiate the multivariate meromorphic decoders,
in order to see the effect of independent coherent errors on each input qubit.
We have simplified to the univariate case for simplicity of exposition.
We also gloss over consideration of the point $z=\infty$ which involves
a change of coordinates such as $z\mapsto 1/z.$


Physicists, especially those doing quantum physics,
blithely write the symbol $i$ in equations, but in terms of 
algebraic number theory this is already problematic: which $i$ do you mean?
There are two solutions of $z^2+1=0$.
In this work, we are seeing how states
are elements of number fields, 
and have names which are polynomials.
And these names are ambiguous up to a group
called the Galois group.
What is the physical or computational significance of this group?

\todo{What about higher dimensional Grassmanians ?
These are absent from textbook quantum, until we discuss
quantum codes and the Knill-Laflamme conditions, but would seem to
naturally exist in the projective theory from day one, except we
have truncated to the zero and one dimensional subspaces.}

\todo{The 2-sphere does not support a topological group structure
which may be related to 
why we don't get a topological monoid structure from the qubit algebra
structures of $ZXW$-calculus.}

By restricting to the surface of the Bloch-Riemann sphere 
we are considering pure states and coherent operations.
Extending to incoherent states requires considering 
anti-holomorphic structure: functions of $\bar{z}$.

Finally, we mention a closely related idea known
as ``geometric quantization''  \cite{Echeverria1999,Blau1992}.
In this work we are treating the projective space
$\Proj(\Hilb)$ as a kind of configuration space for a system.
Going the other way, we can ask, given a complex manifold $\mathcal{M}$,
what extra information is needed to construct a quantum Hilbert space 
for which $\mathcal{M}$ is the classical configuration space?

{\bf Acknowledgements.} The authors thank
Natalie Brown, Shival Dasu, Matthew Burke
and Aleks Kissinger for helpful discussions.

\newpage
\appendix


\section{Proof of Prop.~\ref{lax}}
\label{applax}

We prove that $\PBot:\CVec_\tensor \to \Set_{\bot,\wedge} $
is lax monoidal 
with unitor given by the identity map (the unary $\Sigma$),
and laxator given by $\Sigma$.

Given $\Set_{\times}$ as a strict
monoidal (cartesian) category, we claim that $\PMSet$ is also
a strict monoidal category. For example,
\begin{align*}
((A\cup\{\bot\})\wedge (B\cup\{\bot\}))\wedge (C\cup\{\bot\}) 
&= ((A\times B)\times C)\cup\{\bot\}  \\
&= (A\times (B\times C))\cup\{\bot\}  \\
&= (A\cup\{\bot\})\wedge ((B\cup\{\bot\})\wedge (C\cup\{\bot\}))
\end{align*}
and similarly for the units.
We will also take $FVec_{\tensor}$ as a strict monoidal
category.

We have,
$$
\Proj_\bot(1_\tensor) = \Proj_\bot(\Complex^1) = \{\star\}\cup\{\bot\},
$$
and so we can take the unitor
as the identity map on $\{\star\}\cup\{\bot\}$.

Given vector spaces $V,W\in\CVec_\tensor$ we define the laxator as
the natural transform 
\begin{align*}
\Sigma_{V,W} : \PBot(V)\wedge\PBot(W) & \to \PBot(V\tensor W) \\
    \bot & \mapsto \bot \\
    ([v],\bot) &\mapsto \bot \\
    (\bot,[w]) &\mapsto \bot \\
    ([v],[w]) &\mapsto [v\tensor w]
\end{align*}
\todo{show this is natural}

For associativity we need to show,
given vector spaces $U,V,W\in\CVec$
the following two maps are equal:
$$
\PBot(U)\wedge\PBot(V)\wedge\PBot(W) 
\xlongrightarrow{\Sigma_{U,V}\wedge 1_W} 
\PBot(U\tensor V)\wedge\PBot(W)
\xlongrightarrow{\Sigma_{U\tensor V,W}}
\PBot(U\tensor V\tensor W)
$$
and
$$
\PBot(U)\wedge\PBot(V)\wedge\PBot(W) 
\xlongrightarrow{1_U\wedge\Sigma_{V,W}} 
\PBot(U)\wedge\PBot(V\tensor W)
\xlongrightarrow{\Sigma_{U,V\tensor W}}
\PBot(U\tensor V\tensor W)
$$
Given $u\in U,v\in V,w\in W$
these maps both take $([u],[v],[w])$ 
to $[u\tensor v\tensor w]$ and any appearance
of $\bot$ results in $\bot$.

For unitality, we check that both 
$\Sigma_{\Complex^1,U}$ and $\Sigma_{U,\Complex^1}$ are identity maps
on $\Proj_\bot(\Complex^1\tensor U) = \Proj_\bot(U).$


\section{Proof of Thm.~\ref{wenum}}
\label{appwenum}


Given a CSS code with parameters $[[n,1,d]]$,
stabilizer generators $S_X, S_Z$, logical operators $L_X,L_Z$,
and logical decoder $L^\dag$, we claim that 
\begin{align}\label{eq}
L^\dag \Matrix{x\\y}^{\tensor n}
&= \Matrix{W_{\langle S_X\rangle}(x,y) \\ W_{L_X\langle S_X\rangle}(x,y) }.
\end{align}
When $L_X=X^{\tensor n}$ we have that 
 $W_{L_X\langle S_X\rangle}(x,y) = 
 W_{\langle S_X\rangle}(y,x)$
and the result follows by projectivizing.

To prove the claim (\ref{eq}), 
we write the logical decoder $L^\dag$
as a phase-free ZX diagram called
the \emph{Z-X normal form}~\cite{Kissinger2022}.
By bending wires this becomes 
a $(n+k)$-qubit state, and the 
result follows from~\cite{Kissinger2022} Eq.(5).

\todo{we should spell all this out if we have time}

This proof straightforwardly generalizes in two ways.
Firstly, to CSS codes $[[n,k,d]]$ with $k>1$.
Secondly, we obtain multivariate meromorphic decoder functions
from multivariate versions of the 
weight enumerators called ``full weight enumerators''.


\section{Sage python example code}
\label{appsage}

The following interactive session shows an 
example of using sage via the python interpreter.

\begin{verbatim}
Python 3.10.12 
Type "help", "copyright", "credits" or "license" for more information.
>>> from sage.all_cmdline import *
>>> R = PolynomialRing(ZZ, "z") # polynomial ring in "z" over the integers
>>> z, = R.gens()
>>> F = z**8 + 14*z**4 + 1 # make a polynomial in the z variable
>>> F
z^8 + 14*z^4 + 1
>>> factor(F) # factorize over the integer polynomial ring
(z^4 - 2*z^3 + 2*z^2 + 2*z + 1) * (z^4 + 2*z^3 + 2*z^2 - 2*z + 1)
>>> F(z=1) # evaluate the polynomial at z=1
16
>>> K = PolynomialRing(CyclotomicField(24), "z") # a bigger polynomial ring
>>> KF = K(F) # promote F to this new ring
>>> factor(KF) # KF now completely splits into linear factors
(z - zeta24^4 - zeta24^2) * (z - zeta24^4 + zeta24^2) * 
(z + zeta24^4 - zeta24^2) * (z + zeta24^4 + zeta24^2) * 
(z - zeta24^6 - zeta24^4 + zeta24^2 + 1) * (z - zeta24^6 + zeta24^4 + zeta24^2 - 1) * 
(z + zeta24^6 - zeta24^4 - zeta24^2 + 1) * (z + zeta24^6 + zeta24^4 - zeta24^2 - 1)
>>> f = (z**5 - 5*z) / (5*z**4 - 1) # a meromorphic function
>>> f.parent() # sage knows how to promote fractions
Fraction Field of Univariate Polynomial Ring in z over Integer Ring
>>> diff(f) # differential with respect to the variable z
(5*z^8 + 70*z^4 + 5)/(25*z^8 - 10*z^4 + 1)
>>>


\end{verbatim}

\section{Galois connection}

The 24'th cyclotomic polynomial is $z^8 - z^4 + 1$.
The Galois group for the 24'th cyclotomic field
$\Rational(1^{1/24})\cong\Rational[z]/\{z^8-z^4+1\}$
is $\GL(1,\Integer_{24}) \cong \Integer_2^3.$
This is a three dimensional vector space over the
field $\Field_2$, and so has sixteen
(additive) abelian subgroups:
$1+7+7+1=16$.

\bibliography{refs}{}
\bibliographystyle{abbrv}

\end{document}